\documentclass[preprintnumbers,11pt,onecolumn, letterpaper, accepted=2018-08-17]{quantumarticle}
\pdfoutput=1	

\usepackage{fullpage}
\usepackage{amsfonts, amssymb, amsmath, amsthm}
\usepackage{latexsym}
\usepackage[tracking=smallcaps]{microtype}	
\usepackage{url}
\usepackage{color}
\definecolor{DarkGray}{rgb}{0.1,0.1,0.5}
\usepackage[colorlinks=true,breaklinks, linkcolor=black,citecolor=black,urlcolor=black]{hyperref}	

\usepackage{graphicx}
\usepackage[tight, TABBOTCAP]{subfigure}

\newcommand{\bra}[1]{{\langle#1|}}
\newcommand{\ket}[1]{{|#1\rangle}}

\newcommand{\norm}[1]{{\| #1 \|}}
\newcommand{\bignorm}[1]{{\big\| #1 \big\|}}
\newcommand{\Bignorm}[1]{{\Big\| #1 \Big\|}}

\def\C {{\bf C}}

\def\H {{\mathcal H}}
	\def\L {{\mathcal L}}

\def\S {{\mathcal S}}

\newcommand{\identity}{\ensuremath{\boldsymbol{1}}} 
\newcommand{\Id}{\identity} 

\DeclareMathOperator{\SWAP}{\operatorname{SWAP}}

\newcounter{sprows}

\newlength{\spheight}
\newlength{\spraise}

\newcommand{\comment}[1]{\emph{\color{blue}Comment:\color{black} #1}} 
\newlength{\commentslength}
\newcommand{\comments}[1]{
\hspace{-2\parindent}
\addtolength{\commentslength}{-\commentslength}
\addtolength{\commentslength}{\linewidth}
\addtolength{\commentslength}{-\parindent}
\fcolorbox{blue}{white}{\smallskip\begin{minipage}[c]{\commentslength}
\emph{Comments:}\begin{itemize}#1\end{itemize}\end{minipage}}\bigskip
}
\newcommand{\rem}[1]{}

\newtheorem{theorem}{Theorem}[section]
\newtheorem{lemma}[theorem]{Lemma}
\newtheorem{corollary}[theorem]{Corollary}
\newtheorem{claim}[theorem]{Claim}

\newfont{\subsubsecfnt}{ptmri8t at 11pt}
\renewcommand{\subparagraph}[1]{\smallskip{\subsubsecfnt #1.}}

\newcommand{\eqnref}[1]{\hyperref[#1]{{(\ref*{#1})}}}
\newcommand{\thmref}[1]{\hyperref[#1]{{Theorem~\ref*{#1}}}}
\newcommand{\lemref}[1]{\hyperref[#1]{{Lemma~\ref*{#1}}}}
\newcommand{\corref}[1]{\hyperref[#1]{{Corollary~\ref*{#1}}}}
\newcommand{\defref}[1]{\hyperref[#1]{{Definition~\ref*{#1}}}}
\newcommand{\secref}[1]{\hyperref[#1]{{Section~\ref*{#1}}}}
\newcommand{\figref}[1]{\hyperref[#1]{{Figure~\ref*{#1}}}}
\newcommand{\tabref}[1]{\hyperref[#1]{{Table~\ref*{#1}}}}
\newcommand{\remref}[1]{\hyperref[#1]{{Remark~\ref*{#1}}}}
\newcommand{\appref}[1]{\hyperref[#1]{{Appendix~\ref*{#1}}}}
\newcommand{\claimref}[1]{\hyperref[#1]{{Claim~\ref*{#1}}}}
\newcommand{\factref}[1]{\hyperref[#1]{{Fact~\ref*{#1}}}}
\newcommand{\propref}[1]{\hyperref[#1]{{Proposition~\ref*{#1}}}}
\newcommand{\exampleref}[1]{\hyperref[#1]{{Example~\ref*{#1}}}}
\newcommand{\conjref}[1]{\hyperref[#1]{{Conjecture~\ref*{#1}}}}

\allowdisplaybreaks[1]

\def\COLOR{}
\ifdefined\COLOR
\newcommand{\Alice}[1] {{\color{red} {#1}}}
\newcommand{\Bob}[1] {{\color{blue} {#1}}}
\else
\newcommand{\Alice}[1] {{\color{black} {#1}}}
\newcommand{\Bob}[1] {{\color{black} {#1}}}
\fi

\newcommand{\pst}{\psi_\theta}

\newcommand{\EPRstate}{{\mathrm{EPR}}}

\def \EPRstate {{\mathrm{EPR}}}
\def \omegaopt {{\omega_\mathrm{opt}}}

\renewcommand{\comment}[1]{}\renewcommand{\comments}[1]{}

\begin{document}
\def\compilefullpaper{}

\title{Efficient test for entanglement}
\title{Test for a large amount of entanglement, \newline using few measurements}
\author{Rui Chao}
\affiliation{University of Southern California}
\author{Ben W. Reichardt}
\affiliation{University of Southern California}
\author{Chris Sutherland}
\affiliation{University of Southern California}
\author{Thomas Vidick}
\affiliation{California Institute of Technology}
\date{}

\maketitle

\begin{abstract}
Bell-inequality violations establish that two systems share some quantum entanglement.  We give a simple test to certify that two systems share an asymptotically large amount of entanglement, $n$ EPR states.  The test is efficient: unlike earlier tests that play many games, in sequence or in parallel, our test requires only one or two CHSH games.  One system is directed to play a CHSH game on a random specified qubit~$i$, and the other is told to play games on qubits $\{i, j\}$, without knowing which index is~$i$.  

The test is robust: a success probability within $\delta$ of optimal guarantees distance $O(n^{5/2}\sqrt{\delta})$ from~$n$ EPR states.  However, the test does not tolerate constant~$\delta$; it breaks down for $\delta = \tilde \Omega(1/\sqrt n)$.  We give an adversarial strategy that succeeds within~$\delta$ of the optimum probability using only $\tilde O(\delta^{-2})$ EPR states.  
\end{abstract}

\section{Introduction}

Entanglement separates quantum from classical physics, and is a key source for the power of quantum-mechanical devices.  A natural experimental challenge, to advance our quantum engineering skills and possibly test quantum physics itself, is to generate highly entangled states.  However, for studying the foundations of physics and for cryptographic applications, it is important to take a conservative, adversarial perspective.  A convincing test for entanglement must avoid modeling assumptions---such as that experimental operation X, perhaps shining a laser, implements a unitary $X$ on a particular qubit---or even the assumption that the system is quantum-mechanical at all.  A system may appear to be entangled under simple tests, and yet its behavior might still have an underlying classical explanation.  

Recently, ``loophole-free" Bell-inequality violations have been demonstrated~\cite{Hensen15loopholefreeCHSH, Hensen16loopholetwo, Shalm15loopholeCHSH, Giustina15loopholeCHSH}.  These experiments establish that there exists some entanglement between two experimental systems, ruling out classical local-hidden-variable models, in a very adversarial setting.  Following the rules of the CHSH game~\cite{ClauserHorneShimonyHolt69chshgame}, they ask random questions to the two systems---electrons separated by over a kilometer in~\cite{Hensen15loopholefreeCHSH}---and security is based on the space-like separation of the measurements giving the answers.  This separation prevents any collusion, even with signals moving at the speed of light.  The results are up to high statistical confidence, because unentangled systems could get lucky and pass the tests by chance.  

\medskip

Beyond showing some entanglement, how can one certify that two  systems are \emph{highly} entangled?  

We address this natural question.  We develop a test for $n$ EPR states worth of entanglement, $\tfrac{1}{\sqrt{2^n}}(\ket{00} + \ket{11})^{\otimes n}$.  We will explain this test below, but first let us give some more context.  

A Bell-inequality violation certifies that there is some entanglement; a nearly optimal {violation} can certify that the entanglement is of the specific form of an EPR state~\cite{MagniezMayersMoscaOllivier05selftest, McKagueYangScarani12chshrigidity, ReichardtUngerVazirani13qmip}.  Wu et al.~\cite{WuBancalMcKagueScarani15twoEPRtest} have shown that two CHSH games, played in parallel, can be used to test for $n = 2$ EPR states.  One might reasonably suppose that if playing one or two CHSH games can show some entanglement, then playing many CHSH games should suffice to show lots of entanglement.  It is not so simple.  Different games need not be independent from each other, and complicated dependencies could conceivably allow low-dimensional systems to act like higher-dimensional ones~\cite{ChaoReichardtSutherlandVidick16overlapping}.  

The first test for asymptotically many EPR states was given in~\cite{ReichardtUngerVazirani13qmip}.  The motivation for this test was to develop a scheme for delegated quantum computation, allowing one to {securely} run a quantum circuit across several untrusted devices. The test tolerates polynomially small deviations from the ideal behavior, and in principle can be run with polynomial overhead, but it is severely impractical.  The ideal systems need to share much more entanglement, {$N = n^{O(1)} \gg n$} EPR states, than is being certified.  The test is based on playing $N$ CHSH games \emph{sequentially}.  This is highly constraining: the answer to one game must be given before the question for the next game, and yet the whole sequence of questions and answers must be space-like separated from one system to the other.  Finally, the test tolerates only an inverse-polynomial error rate in the systems.  

McKague~\cite{McKague15parallelselftesting} gave a much-improved test for $n$ EPR states.  His test uses only~$n$ EPR states; none are lost in the analysis.  They are all measured, and the measurements must still be space-like separated from one system to the other, but they can be performed in parallel.  The final distance from $\tfrac{1}{\sqrt{2^n}}(\ket{00} + \ket{11})^{\otimes n}$, in Euclidean norm, is $O(\sqrt{n} \, \delta^{1/8})$, where $\delta$ is an upper bound on the error for each of $\tilde O(n)$ test settings.  Thus to achieve error $\epsilon$, one should set $\delta \sim \epsilon^8 / n^4$, so each EPR state and one-qubit measurement should have error $\delta / n \sim \epsilon^8 / n^5$.  Then $\tilde O(n \cdot 1/\delta^2) = \tilde O(n^9 / \epsilon^{16})$ experiments suffice to establish, say, $99\%$ statistical confidence.  

\medskip

Our test is simpler.  Let $i$ and $j$ be uniformly random distinct indices between $1$ and~$n$.  \mbox{Conceptually}, the key component of the test amounts to giving the first system $\{i, j\}$ and the second system~$i$, or vice versa; ask each system to play CHSH games on the specified qubits.  (See \secref{s:protocol} for a complete description of the protocol, which involves a couple other sub-tests.)  Thus most of the EPR states are not destructively measured, so they can be used later.  

We show that passing our test with probability within $\delta$ of the optimal value implies that the systems are $O(n^{5/2} \sqrt \delta)$ close to $n$ EPR states.  Therefore to achieve a final error of~$\epsilon$, the error on individual EPR states and one-qubit measurements can be $\delta \sim \epsilon^2 / n^5$.  To achieve a given statistical confidence, $O(1/\delta^2) = O(n^{10} / \epsilon^4)$ experiments suffice.\footnote{Avoiding a Markov inequality in the proof can improve this to $O(n^9 / \epsilon^4)$.}  This is obviously quite inefficient; the main novelty and advantage of our protocol is its simple form.  

\medskip

The intuition behind our test is as follows.  Call the two systems Alice and Bob.  If they truly share $n$ EPR states, then it is easy for them to pass the test.  In general, Alice and Bob can do anything at all.  For example, they might try to devise some scheme that uses only $n - 1$ EPR states.  The reason that $n$ EPR states are required is that when Bob is asked $\{i, j\}$ he has to select two qubits, one of which will be tested for entanglement with Alice---and he does not know which one.  Monogamy of entanglement prevents both qubits from simultaneously forming EPR states with Alice's qubit.  Therefore (skipping over some arguments), Alice and Bob must share two EPR states.  Chaining together pairwise statements like this establishes that they share $n$ EPR states.  

Note that this test requires very little processing to test a high-dimensional quantum state, as the number of possible operations to be performed scales quadratically with the number of qubits tested.  This contrasts with the exponential number of possibilities required for parallel self-tests.  Possibly a further value is that the test is largely nondestructive; only two of the~$n$ EPR states are measured in the test, leaving $n - 2$ verified EPR states for other uses.  

The practical usefulness of the test, however, depends on its tolerance to imperfections. We characterize a test as ``robust'' if the guaranteed distance to the target state scales polynomially with the sub-optimality incurred by, say, instrumental imperfection. Precisely, we say that a test has ``robustness $f(n,\delta)$'' if a success probability that is $\delta$-close to optimal guarantees a distance $f(n,\delta)$ to the target state. For example, the test from~\cite{ReichardtUngerVazirani13qmip} is robust, and our test is robust as well: the dependence on errors grows only polynomially in~$n$, not exponentially. This is not to say that the test can readily be implemented experimentally.  Might a stronger analysis show, for example, that if Alice and Bob pass the test with $99\%$ probability, then they must share a state with $99\%$ fidelity to $\tfrac{1}{\sqrt{2^n}}(\ket{00} + \ket{11})^{\otimes n}$?  No.  We give an explicit construction by which Alice and Bob can use logarithmically many EPR states, $O\big(\tfrac{1}{\delta^2} \log n\big)$, to pass the test with probability within~$\delta$ of the optimal value.  This implies that the analysis breaks down for $\delta = \Omega(\sqrt{(\log n)/n})$.  The test does not tolerate constant error rates.  (A similar limitation will likely hold for other tests that consider only two qubits at a time~\cite{ChaoReichardtSutherlandVidick16overlapping}.)  

\smallskip

Our strategy for exhibiting $n$ qubits is based on a technique first used in~\cite{KempeKobayashiMatsumotoTonerVidick07qmip}.  We introduce two main ideas.  The first consists in combining $2n (n-1)$ pairwise approximate commutation relations $[P_i,Q_j] \ket \psi \approx 0$, for $P, Q \in \{X, Z\}$ and $i \neq j$, with the exact commutation that always exists between an operator on Alice's Hilbert space~$\H_\Alice{A}$ and an operator on Bob's space~$\H_\Bob{B}$.  We crucially use the fact that, as a consequence of the CHSH test, an operator $P_i$~or~$Q_j$ can be ``pulled'' from $\H_\Alice{A}$ to $\H_\Bob{B}$, at which point it commutes with \emph{any} operator on~$\H_\Alice{A}$.  This operation allows us to control the error blow-up that would become unmanageable if we only had access to a single space~$\H$.  The second idea is to enforce the desired approximate commutation $[P_i,Q_j] \ket \psi \approx 0$ through the use of an intermediate set of operators, $P_i^{\{i,j\}}$~and~$Q_j^{\{i,j\}}$, that are obtained by adding an additional test in which the CHSH test is executed on a pair $\{i,j\}$ of (purported) qubits.  The use of a ``dummy question'' was first introduced in~\cite{ItoKobayashiMatsumoto09oracularization} for the purposes of inducing approximate commutation, and we show that it can be effectively leveraged in our context as well.  

Although our main result is based on the CHSH test and the use of EPR states, we show in \appref{s:psiselftest} that the result can be generalized to separate qubits via any two-qubit state entangled across $\H_\Alice{A} \otimes \H_\Bob{B}$.  It is an open question whether similar results can be attained based on higher-dimensional partially entangled states.  

\smallskip

Recently, several other entanglement tests have been given.  Coudron and Natarajan~\cite{CoudronNatarajan16rigidmagicsquare} and Coladangelo~\cite{Coladangelo16parallelchsh} study parallel repetition of the Magic Square game~\cite{Mermin90magicsquare, Peres90magicsquare}, achieving robustness of $O(n^2 \sqrt \delta)$ and $O(n^{3/2} \sqrt \delta)$, respectively.  The optimal success probability for this game is~$1$, which allows for a more efficient analysis in comparison to the CHSH game. They also require fewer experiments to achieve certain statistical confidence.  Natarajan and Vidick~\cite{NatarajanVidick17lineartest} leverage a quantum version of the linearity test by Blum et al.~\cite{BLRtest}, and {obtain} robustness that is independent of $n$. Their test is rather complex and requires the verifier to choose among exponentially many possible questions. This was recently improved~\cite{NatarajanVidick2018lowdegree} to a test that achieves simultaneously poly$(n)$ questions and constant robustness ($O(\sqrt \delta)$).  Subsequently to our work, Ostrev and Vidick~\cite{OstrevVidick16entanglement} apply our \thmref{t:EPRstabilizersfrommath} to analyze a simpler XOR game than ours, resulting in similar number of possible questions, but slightly weaker robustness guarantees. It is an open question to determine the best trade-offs in terms of robustness, number of questions, and number of EPR pairs tested.  Robust protocols certifying a large amount of entanglement, without explicitly certifying that the state must be (close to) maximally entangled, are provided in~\cite{ColadangeloStark17lineargame,RotemYuen17highdim,RotemBancal17didistill}.

\subsection{Notation}

Denote the EPR state by $\ket{\EPRstate} = \frac{1}{\sqrt 2}(\ket{00} + \ket{11})$.  Let $[n] = \{1, \ldots, n\}$.  The Pauli matrices are $I = \big(\begin{smallmatrix}1&0\\0&1\end{smallmatrix}\big)$, $\sigma^x = \big(\begin{smallmatrix}0&1\\1&0\end{smallmatrix}\big)$, $\sigma^y = \big(\begin{smallmatrix}0&-i\\i&0\end{smallmatrix}\big)$ and $\sigma^z = \big(\begin{smallmatrix}1&0\\0&-1\end{smallmatrix}\big)$.  On higher-dimensional spaces, we will denote the identity operator by~$\identity$.  

We consider Hilbert spaces $\H_\Alice{A}, \H_\Alice{A'}, \H_\Bob{B}, \H_\Bob{B'}$, etc., which for clarity are labeled by the register $\Alice{A}, \Alice{A'}, \Bob{B}$ or $\Bob{B'}$ that contains the quantum state whose state space they represent.  Each such space is assumed to be finite dimensional.  We will use the letter $D$ as a variable which ranges over a subset of $\{\Alice{A}, \Alice{A'}, \Bob{B}, \Bob{B'}, \ldots\}$ that will be clear from context.

\section{State-dependent separation of $n$ EPR states} \label{s:EPRselftest}

In this section we prove a separation theorem for qubits that satisfy certain state-dependent \mbox{commutation} relations. In \secref{s:protocol} we provide a simple protocol that can establish the state-dependent assumptions in an experimental setting. 

\begin{theorem} \label{t:EPRstabilizersfrommath}
Let $\ket \psi \in \H_\Alice{A} \otimes \H_\Bob{B}$, with $\norm{\ket \psi} = 1$.  Assume that for $j \in [n]$ we are given reflections 
\begin{equation*}
(X_j)_\Alice{A}, (Z_j)_\Alice{A} \in \L(\H_\Alice{A})
\quad \text{and} \quad
(X_j)_\Bob{B}, (Z_j)_\Bob{B} \in \L(\H_\Bob{B})
\end{equation*}
that for $D$ either $\Alice{A}$ or~$\Bob{B}$, all $i \neq j$, and $P, Q$ either $X$ or~$Z$, satisfy $\{(X_j)_D, (Z_j)_D\} = 0$ and 
\begin{align*}
\bignorm{ [P_i, Q_j]_D \ket \psi } &\leq \epsilon \\
\bignorm{ (P_j)_\Alice{A} \otimes (P_j)_\Bob{B} \ket \psi - \ket \psi } &\leq \epsilon
 \enspace .
\end{align*}

Let 
\begin{equation*}
\ket{\psi'} = \ket \psi \otimes \ket{\EPRstate}^{\otimes n}_\Alice{A'} \otimes \ket{\EPRstate}^{\otimes n}_\Bob{B'} \in \H_\Alice{A} \otimes (\C^2)^{\otimes 2n}_\Alice{A'} \otimes \H_\Bob{B} \otimes (\C^2)^{\otimes 2n}_\Bob{B'}
 \enspace .
\end{equation*}
Then for $D \in \{\Alice{A}, \Bob{B}\}$, there exist reflections $X_1', Z_1', \ldots, X_n', Z_n'$ on $\H_D \otimes (\C^2)^{\otimes 2n}_{D'}$, with $\{X_j', Z_j'\} = 0$, $[P_i', Q_j'] = 0$ for $i \neq j$, and $\bignorm{ \big( P_j' - P_j \otimes \identity_{D'} \big) \ket{\psi'} } = O(n \epsilon)$, and furthermore satisfying\footnote{For notational sanity, we suppress some subscripts $D$ or $D'$.  Of course operators $P_j'$ acting on $\Alice{AA'}$ and $\Bob{BB'}$ are in general different.}  
\begin{equation} \label{e:eprccl}
\bignorm{ (P_j')_\Alice{AA'} \otimes (P_j')_\Bob{BB'} \ket{\psi'} - \ket{\psi'} } = O(n \epsilon)
 \enspace .
\end{equation}
\end{theorem}

In words, the assumption of the theorem is that Alice and Bob have $n$ overlapping qubits, with Pauli operators on different qubits nearly commuting on $\ket \psi$, and such that $\ket \psi$ is nearly stabilized by $(X_j \otimes X_j)_{\Alice{A} \Bob{B}}$ and $(Z_j \otimes Z_j)_{\Alice{A} \Bob{B}}$ for each~$j$.  Note that these assumptions are satisfied by $\ket \psi = \ket{\EPRstate}^{\otimes n}_{\Alice{A}\Bob{B}}$ and $X_j = \sigma_j^x$, $Z_j = \sigma_j^z$, with $\epsilon = 0$.  

The conclusion of the theorem is that, up to local isometries (that introduce the ancillary states $\ket{\EPRstate}^{\otimes n}_\Alice{A'}$ and $\ket{\EPRstate}^{\otimes n}_\Bob{B'}$), there exist $n$ non-overlapping qubits for each of Alice and Bob, with associated Pauli operators $(P_j')_{DD'}$, that on the extended state $\ket{\psi'}$ are nearly the same as the original qubits' Pauli operators, and such that $\ket{\psi'}$ is nearly stabilized by $(X_j' \otimes X_j')_{\Alice{AA'} \Bob{BB'}}$ and $(Z_j' \otimes Z_j')_{\Alice{AA'} \Bob{BB'}}$.  By~\cite[Theorem~2.3]{ChaoReichardtSutherlandVidick16overlapping}, there are local changes of basis under which these new qubits are in tensor product.  Since $\ket{\EPRstate}$ is the unique pure state stabilized by both $\sigma^x \otimes \sigma^x$ and $\sigma^z \otimes \sigma^z$, it follows that $\ket{\psi'}$ is close to $n$ EPR states between Alice and Bob, plus an extra ancillary state, with $X_j', Z_j'$ the standard Pauli operators on one half of the $j$th EPR state.  We state this as a corollary to \thmref{t:EPRstabilizersfrommath}: 

\def\ancilla {{\mathrm{extra}}}

\begin{corollary} \label{t:EPRpairsfromstabilizers}
Under the assumptions of \thmref{t:EPRstabilizersfrommath}, there are unitaries $U_D : \H_D \otimes (\C^2)^{\otimes 2n}_{D'} \rightarrow (\C^2)^{\otimes n}_D \otimes \hat \H_D$, for $D\in\{\Alice{A},\Bob{B}\}$, and a state $\ket{\ancilla} \in \hat \H_\Alice{A} \otimes \hat \H_\Bob{B}$ such that 
\begin{align*}
\bignorm{ U_\Alice{A} \otimes U_\Bob{B} \ket{\psi'} - \ket{\EPRstate}_{\Alice{A}\Bob{B}}^{\otimes n} \otimes \ket{\ancilla} } &= O(n^{3/2} \epsilon) \\
\norm{ (U X_j U^\dagger - \sigma^x_j \otimes \identity_{\hat \H})_D \ket{\EPRstate}_{\Alice{A}\Bob{B}}^{\otimes n} \otimes \ket{\ancilla} } &= O(n^{3/2} \epsilon) \\
\norm{ (U Z_j U^\dagger - \sigma^z_j \otimes \identity_{\hat \H})_D \ket{\EPRstate}_{\Alice{A}\Bob{B}}^{\otimes n} \otimes \ket{\ancilla} } &= O(n^{3/2} \epsilon) 
 \enspace .
\end{align*}
\end{corollary}

\begin{proof}
Theorem~2.3 in~\cite{ChaoReichardtSutherlandVidick16overlapping} gives an isomorphism from $\H'_D$ to $(\C^2)^{\otimes n}_D \otimes \hat \H_D$ under which the $(X_j')_D$ and $(Z_j')_D$ promised by \thmref{t:EPRstabilizersfrommath} are simply $\sigma^x_j \otimes \identity_{\hat \H_D}$ and $\sigma^z_j \otimes \identity_{\hat \H_D}$, respectively.  The first claim below then shows that Eq.~\eqnref{e:eprccl} implies that under this isometry $\ket{\psi'}$ must be close to an EPR state on the $j$th qubits, for every $j \in [n]$.  

\begin{claim} \label{t:EPRstatestabilized}
There exists a universal constant~$c$ such that: if $\ket \phi \in \C^2 \otimes \C^2 \otimes \H$ is a state with $\norm{\sigma^x_1 \sigma^x_2 \ket \phi - \ket \phi}, \norm{\sigma^z_1 \sigma^z_2 \ket \phi - \ket \phi} \leq \delta$, then for some state $\ket{\phi'} \in \H$, 
\begin{equation*}
\bignorm{ \ket \phi - \ket{\EPRstate} \otimes \ket{\phi'} } \leq c \, \delta
 \enspace .
\end{equation*}
\end{claim}

\begin{proof}
Expand $\ket \phi = \sum_{a, b \in \{0,1\}} \ket{a,b} \ket{\phi_{a,b}}$.  Then $\norm{\sigma^z_1 \sigma^z_2 \ket \phi - \ket \phi}^2 = 4 (\norm{\ket{\phi_{01}}}^2 + \norm{\ket{\phi_{10}}}^2)$ and $\norm{\sigma^x_1 \sigma^x_2 \ket \phi - \ket \phi}^2 = \sum_{a,b} \norm{\ket{\phi_{ab}} - \ket{\phi_{\bar a \bar b}}}^2$.  The conclusion follows, for $\ket{\phi'} = \ket{\phi_{00}} / \norm{\ket{\phi_{00}}}$.  
\end{proof}

The next claim completes the proof by showing that if $\ket{\psi'}$ is $\delta$-close to an EPR state on each of $n$ registers, then it is $O(\sqrt{n} \delta)$-close to $n$ EPR states. 

\begin{claim} \label{t:manyEPRstatesstabilized}
If $\ket \psi \in (\C^d)^{\otimes n} \otimes \H$ and $\ket{\phi_1}, \ldots, \ket{\phi_n} \in (\C^d)^{\otimes (n-1)} \otimes \H$ are states such that for all $j \in [n]$, $\norm{\ket \psi - \ket{1}_j \ket{\phi_j}_{-j}} \leq \delta$, then there exists a state $\ket \varphi \in \H$ such that 
\begin{equation*}
\bignorm{ \ket \psi - \ket{1}^{\otimes n} \ket \varphi } \leq \sqrt{2 n} \delta
 \enspace .
\end{equation*}
\end{claim}

\begin{proof}
Expand $\ket \psi = \sum_{x \in [d]^n} \ket x \ket{\alpha_x}$.  For each $j$, the state $\ket{\phi_j}$ that minimizes $\norm{\ket \psi - \ket{1}_j \ket{\phi_j}_{-j}}$ is given by $\ket{\hat \phi_j} / \norm{\hat \phi_j}$, where $\ket{\hat \phi_j} = \sum_{x : x_j = 1} \ket{x_{-j}} \ket{\alpha_x}$.  Then 
\begin{align*}
\delta^2 
&\geq 
\Bignorm{ \ket \psi - \ket{1}_j \tfrac{1}{\norm{\hat \phi_j}} \ket{\hat \phi_j}_{-j} }^2 \\
&= 
\Bignorm{ \ket{1}_j \Big(1 - \tfrac{1}{\norm{\hat \phi_j}}\Big) \ket{\hat \phi_j} + \big(\ket \psi - \ket{1}_j \ket{\hat \phi_j}\big) }^2 \\
&= 
\big(1 - \norm{\hat \phi_j}\big)^2 + \sum_{x : x_j \neq 1} \norm{\ket{\alpha_x}}^2 \\
&= 
2 - 2 \norm{\ket{\hat \phi_j}}
 \enspace .
\end{align*}
Thus for each $j$, we have $\sum_{x : x_j \neq 1} \norm{\alpha_x}^2 \leq \delta^2$.  Under the constraint $\sum_x \norm{\alpha_x}^2 = 1$ and by a union bound, we have $\norm{\alpha_{1^n}}^2 \geq 1 - n \delta^2$.  Provided that $n \delta^2 < 1$, $\norm{\alpha_{1^n}} > 0$.  

Letting $\ket \varphi = \ket{\alpha_{1^n}} / \norm{\alpha_{1^n}}$, we get 
\begin{align*}
\bignorm{ \ket \psi - \ket{1}^{\otimes n} \ket \varphi }^2
&= 
2 - 2 \norm{\alpha_{1^n}} \\
&\leq 
2 - 2 \sqrt{1 - n \delta^2} \\
&\leq 2 n \delta^2
 \enspace .
\end{align*}
This gives the bound $\norm{\ket \psi - \ket{1}^{\otimes n} \ket \varphi} \leq \sqrt{2 n} \delta$ provided $\delta < 1 / \sqrt n$.  If $\delta \geq 1 / \sqrt n$, it is still easy to choose a state $\ket \varphi$ so $\norm{\ket \psi - \ket{1}^{\otimes n} \ket \varphi} \leq \sqrt 2 \leq \sqrt{2 n} \delta$.  
\end{proof}

\corref{t:EPRpairsfromstabilizers} follows.  
\end{proof}

\begin{proof}[Proof of \thmref{t:EPRstabilizersfrommath}]
The idea for the proof is as follows.  There are $n$ overlapping qubits on Alice's side, and another $n$ overlapping qubits on Bob's side.  The assumptions of the theorem imply that for each~$j$, the joint state of Alice and Bob's qubits~$j$ is close to an EPR state.  For the analysis, we sequentially swap in fresh qubits in order to force a tensor-product structure.  Since there is an underlying state, $\ket \psi$, we need to specify a state for the fresh qubits.  The natural choice is the maximally mixed state, $\tfrac12 I$,  because locally this looks the same as half of an EPR state, and thus Alice and Bob's local operators cannot tell that we have changed the state from underneath them.  On Alice's side, therefore, the local isometry adds $n-1$ EPR states (entirely on Alice's side, not crossing between Alice and Bob).  Half of each EPR state is unused, and we swap in the other half, which looks maximally mixed.  

We proceed with the details.  For each~$j$, let $(\S_j)_\Alice{AA'}$ be the operator on $\H_\Alice{A} \otimes (\C^2)^{\otimes 2n}_\Alice{A'}$ that swaps Alice's $j$th qubit, defined by $(X_j)_\Alice{A}, (Z_j)_\Alice{A}$, with half of the $j$th appended EPR state: 
\begin{equation*}
(\S_j)_\Alice{AA'} = \frac12 \big( \identity \otimes \identity + X_j \otimes \sigma^x_{2j-1} + Z_j \otimes \sigma^z_{2j-1} + i (X_j Z_j) \otimes \sigma^y_{2j-1} \big)_\Alice{AA'}
 \enspace .
\end{equation*}
For $P$ either $X$ or~$Z$, define $(P_j')_\Alice{AA'}$ by 
\begin{equation*}
P_j' = (\S_1 \cdots S_{j-1}) \, (P_j \otimes \identity) \, (\S_{j-1} \cdots \S_1)
 \enspace .
\end{equation*}
Then for $P, Q \in \{X, Z\}$ and $i < j$, 
\begin{equation*}
\norm{[P_i', Q_j']} 
= \norm{[\S_i (P_i \otimes \identity) \S_i, \S_{i+1} \cdots \S_{j-1} (Q_j \otimes \identity) \S_{j-1} \cdots \S_{i+1}]} 
= 0
 \enspace ,
\end{equation*}
since $\S_i (P_i \otimes \identity) \S_i$ is a Pauli on the $(2i-1)$th added qubit, on which the other term has no dependence.  
Let 
\begin{equation*}
(\S_j')_{\Bob{B}\Alice{A'}} = \frac12 \big( \identity \otimes \identity + X_j \otimes \sigma^x_{2j} + Z_j \otimes \sigma^z_{2j} + i (X_j Z_j) \otimes \sigma^y_{2j} \big)_{\Bob{B}\Alice{A'}}
\end{equation*}
be the operator that swaps Bob's $j$th qubit with the other half of Alice's $j$th appended EPR state.  
Then we compute 
\begin{align*}
(P_j')_\Alice{AA'} \ket{\psi'}
&= (\S_1 \cdots \S_{j-1} \, P_j \, \S_{j-1} \cdots \S_1)_\Alice{AA'} \ket{\psi'} \\
&\approx_\epsilon (\S_1 \cdots \S_{j-1} \, P_j \, \S_{j-1})_\Alice{AA'} (\S_1' \cdots \S_{j-2}')_{\Bob{B}\Alice{A'}} \ket{\psi'} && \text{(\lemref{t:pullswapfromAlicetoBob})} \\
&\approx_\epsilon (\S_1 \cdots \S_{j-2} \, P_j)_\Alice{AA'} (\S_1' \cdots \S_{j-2}')_{\Bob{B}\Alice{A'}} \ket{\psi'} && \text{(\lemref{t:pauliswapcommuteondifferentqubits})} \\
&\approx \cdots \\
&\approx_\epsilon (P_j)_\Alice{A} \ket{\psi'}
 \enspace ,
\end{align*}
where each approximation holds up to order $\epsilon$ in Euclidean norm.  The first approximation, pulling $\S_{j-2} \cdots \S_1$ over to Bob's side, is by \lemref{t:pullswapfromAlicetoBob} below.  The next line, commuting $\S_{j-1}$ past $P_j$ on $\ket{\psi'}$, is by \lemref{t:pauliswapcommuteondifferentqubits}.  The argument continues by pulling one $\S_i'$ at a time back to Alice's side, where it can be commuted past $P_j$.  Overall, the error is $\norm{(P_j')_\Alice{AA'} \ket{\psi'} - (P_j)_\Alice{A} \ket{\psi'}} = O(n \epsilon)$.  

The same argument can be repeated for swap operators $(\S_j)_\Bob{BB'}$ defined on Bob's side.  Thus $\norm{(P_j')_\Bob{BB'} \ket{\psi'} - (P_j)_\Bob{B} \ket{\psi'}} = O(n \epsilon)$, and so in particular 
\begin{align*}
\bignorm{ (P_j')_\Alice{AA'} \otimes (P_j')_\Bob{BB'} \ket{\psi'} - \ket{\psi'} } 
&\leq \bignorm{ (P_j' \otimes P_j')_{\Alice{AA'}\Bob{BB'}} \ket{\psi'} - (P_j \otimes P_j)_{\Alice{A}\Bob{B}} \ket{\psi'} } \\
&\quad + \bignorm{ (P_j \otimes P_j)_{\Alice{A}\Bob{B}} \ket{\psi'} - \ket{\psi'} } \\
&= O(n \epsilon)
 \enspace ,
\end{align*}
as claimed.  
\end{proof}

Lemmas~\ref{t:pullswapfromAlicetoBob} and~\ref{t:pauliswapcommuteondifferentqubits} are simple calculations based on the triangle inequality.  

\begin{lemma} \label{t:pullswapfromAlicetoBob}
$\bignorm{(\S_i)_\Alice{AA'} \ket{\psi'} - (\S_i')_{\Bob{B}\Alice{A'}} \ket{\psi'}} = O(\epsilon)$.  
\end{lemma}

\begin{proof}
\begin{figure}
\centering
\includegraphics[scale=.12]{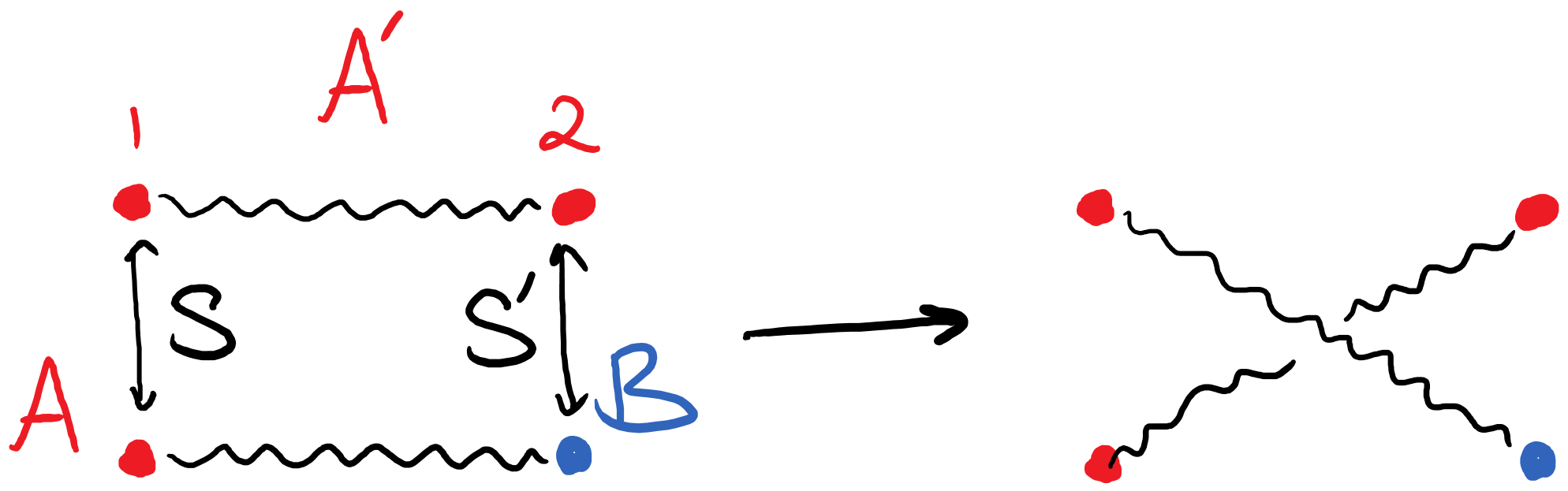}
\caption{Two ways of swapping halves of EPR states, that both lead to the same result.} \label{f:swappinghalvesoftwoEPRstates}
\end{figure}

If we had $(X_i)_\Alice{A} (X_i)_\Bob{B} \ket \psi = \ket \psi = (Z_i)_\Alice{A} (Z_i)_\Bob{B} \ket \psi$, then we would have $(\S_i)_\Alice{AA'} \ket{\psi'} = (\S_i')_{\Bob{B}\Alice{A'}} \ket{\psi'}$ exactly, since the EPR state is uniquely stabilized by both $\sigma^x \otimes \sigma^x$ and $\sigma^z \otimes \sigma^z$.  See \figref{f:swappinghalvesoftwoEPRstates}.  Thus $\bignorm{ (\S_i)_\Alice{AA'} \ket{\psi'} - (\S_i')_{\Bob{B}\Alice{A'}} \ket{\psi'} } \leq \bignorm{\ket{\psi'} - \ket{\psi''}} + \bignorm{\ket{\psi''}-(\S_i)_\Alice{AA'}(\S_i')_{\Bob{B}\Alice{A'}}\ket{\psi'}} \leq 2 c \epsilon$, where $\ket{\psi''}$ restricted on Alice's and Bob's $i$th qubits is an EPR state and $c$ is the constant from \claimref{t:EPRstatestabilized}.
\end{proof}

\begin{lemma} \label{t:pauliswapcommuteondifferentqubits}
For $i \neq j$, $\bignorm{ [\S_i, P_j \otimes \identity]_\Alice{AA'} \ket{\psi'} } = O(\epsilon)$.  
\end{lemma}

\begin{proof}
Use the definition of $\S_i$ in terms of $X_i$ and $Z_i$, which by assumption nearly commute with $P_j$ on $\ket \psi$.  
\end{proof}

\section{A protocol for testing $n$ qubits} \label{s:protocol}

It remains to explain how the mathematical assumptions of \thmref{t:EPRstabilizersfrommath} can be established in an experiment.  The assumption that $(P_j)_\Alice{A} \otimes (P_j)_\Bob{B} \ket \psi \approx \ket \psi$ is straightforward to establish by testing for an EPR state using the CHSH game; see \secref{s:CHSHgame} below. The EPR state {$\ket{\EPRstate} = \tfrac{1}{\sqrt 2}(\ket{00} + \ket{11})$} is an eigenvalue-one eigenvector of both $\sigma^x \otimes \sigma^x$ and $\sigma^z \otimes \sigma^z$.  

The assumption that $[P_i, Q_j]_D \ket \psi \approx 0$ for any $i \neq j$ seems more difficult to establish, because it involves $(P_i)_D$ acting on $(Q_j)_D \ket \psi$, not just on $\ket \psi$ directly.  The main idea to obtain constraints on quantities such as $(P_i Q_j)_D \ket \psi$ is to use both players, and design a test where the second player performs a measurement that relates to both $P_i$ and $Q_j$.  In \secref{s:EPRpairprotocol} we present a simple protocol for which the soundness analysis uses this argument.

\subsection{The CHSH game} \label{s:CHSHgame}

In \secref{s:EPRpairprotocol} below we give a protocol that can experimentally establish the assumptions of \thmref{t:EPRstabilizersfrommath}.  This protocol is built on the Clauser-Horne-Shimony-Holt (CHSH) game~\cite{ClauserHorneShimonyHolt69chshgame}.  Before presenting our protocol, let us review the CHSH game.  

The CHSH game involves a verifier and two players.  The verifier chooses uniformly random bits $a, b \in \{0,1\}$.  As shown in \figref{f:chshgame}, she sends $a$ to the first player, and~$b$ to the second player.  The players share a quantum state but cannot communicate.  Each player applies a two-outcome measurement, depending on her input.  They return their respective outcomes, $x$~and~$y$ in $\{1, -1\}$, to the verifier.  The players win if $x y = (-1)^{a b}$.  

\begin{figure}
\centering
\includegraphics[scale=1]{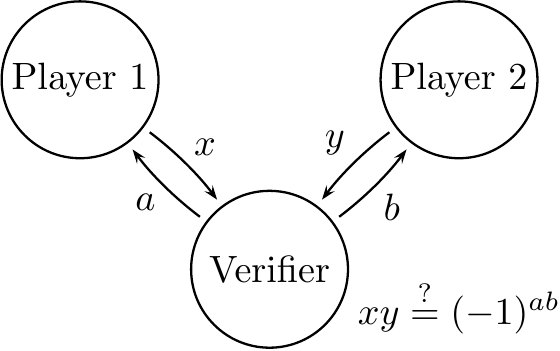}
\caption{The CHSH game.  The messages $a, b$ are in $\{0,1\}$, and $x, y \in \{1, -1\}$.} \label{f:chshgame}
\end{figure}

A general strategy for the CHSH game consists of Hilbert spaces $\H_1, \H_2$, a state $\ket \psi \in \H_1 \otimes \H_2$, and reflections $(R_0)_D, (R_1)_D \in \L(\H_D)$ for $D \in \{1, 2\}$.  On receiving $c \in \{0,1\}$, each player measures $R_c$ and returns the outcome.  Using this strategy, the players win with probability 
\begin{equation*}
\frac12 + \frac18 \bra \psi \big( R_0 \otimes R_0 + R_0 \otimes R_1 + R_1 \otimes R_0 - R_1 \otimes R_1 \big) \ket \psi
 \enspace .
\end{equation*}

The maximum probability of winning is $\cos^2 \tfrac{\pi}{8} \approx 0.85$.  A strategy achieving this probability uses $\ket \psi = \ket{\EPRstate} \in \C^2 \otimes \C^2$, and the reflections $(R_0)_1 = \sigma^z$, $(R_1)_1 = \sigma^x$ for player~$1$, and $(R_0)_2 = H = \tfrac{1}{\sqrt 2}\big(\begin{smallmatrix}1&1\\1&-1\end{smallmatrix}\big)$, $(R_1)_2 = - \sigma^x H \sigma^x$ for player~$2$.  

In fact, up to local changes of basis and choice of an ancillary state, this is the only optimal strategy~\cite{McKagueYangScarani12chshrigidity, ReichardtUngerVazirani12qmip}.  We will use: 

\begin{lemma}[{\cite[Lemma~4.2]{ReichardtUngerVazirani12qmip}}] \label{t:eprlemma}
Consider a strategy for the CHSH game that wins with probability at least $\cos^2 \tfrac\pi8 - \gamma$ for some $\gamma > 0$.  Assume that for $D \in \{1, 2\}$, the reflections $(R_c)_D$ each have both $\pm 1$ eigenspaces of the same dimension.  Then the following properties hold: 
\begin{itemize}
\item
For each $D \in \{1, 2\}$, there is an isomorphism between player~$D$'s Hilbert space and $\C^2 \otimes \hat \H_D$, under which $(R_0)_D = \sigma^z \otimes \identity$ and $\bignorm{(R_1 - \sigma^x \otimes \identity)_D \ket \psi} = O(\sqrt \gamma)$.  The isomorphism depends only on $(R_0)_D$ and $(R_1)_D$.  
\item 
There exists a unit vector $\ket{\hat \psi} \in \hat \H_1 \otimes \hat \H_2$ such that for $G = \exp(-i \tfrac\pi8 \sigma^y)$, 
\begin{equation*}
\bignorm{\ket \psi - (I \otimes (H G) \ket{\EPRstate}) \ket{\hat \psi}} = O(\sqrt \gamma)
 \enspace .
\end{equation*}
\end{itemize}
\end{lemma}

Thus the CHSH game provides a test to establish that both players measure a single qubit in a shared EPR state.  Our protocol for testing $n$ qubits uses many embedded CHSH games.

\subsection{A protocol for testing $n$ qubits} \label{s:EPRpairprotocol}

In this section, we describe and analyze a simple protocol, between a verifier and two players, Alice and Bob, that can experimentally establish the assumptions of \thmref{t:EPRstabilizersfrommath}.  Recall, these assumptions are that Alice and Bob each have $n$ overlapping qubits, such that the joint two-qubit state of their $j$th qubits is close to an EPR state, for every~$j$, and with Pauli operators on different qubits nearly commuting on $\ket \psi$.  Intuitively, approximate EPR states are easily established using the CHSH game.  The last commutation assumption is harder to establish.  

For example, one possible protocol might have the verifier choose a random index $j \in [n]$, send $j$ to both players, and ask them to play a CHSH game on the $j$th of $n$ shared EPR states.  Honest players, who indeed shared $\ket{\EPRstate}{}^{\otimes n}$, could pass this protocol.  However, dishonest players could also pass this protocol by ignoring~$j$ and using always the same EPR state to play all CHSH games.  In this case, each player's $n$ qubits would all be the same qubit, so the protocol is inadequate.  

Our protocol overcomes this flaw.  It is as follows: 

{ \noindent \hrulefill \\
\centering \textbf{Protocol for testing $n$ qubits} \\ } \smallskip

A verifier interacts with two players, Alice and Bob.  With equal probabilities $1/3$, the verifier runs one of the following sub-protocols: 
\begin{enumerate}
\item Choose $i \in [n]$ and $a \in \{0,1\}$ uniformly at random.  Send $(i, a)$ to each player.  Alice and Bob reply, respectively, with $x, y \in \{1, -1\}$.  Accept if and only if $x = y$.  
\item Choose $i, j \in [n]$, $i \neq j$, and bits $a, b, c \in \{0, 1\}$ uniformly at random.  Send Alice $(i, a)$, and send Bob $\big((i, b), (j, c) \big)$ if $i < j$, or $\big((j, c), (i, b)\big)$ if $j < i$.  Alice returns $x \in \{\pm 1\}$.  Bob returns $y$ and $y'$ in $\{\pm 1\}$, ordered to correspond respectively to $(i, b)$ and $(j, c)$.  Accept if and only if $x y = (-1)^{a b}$.  
\item The same as sub-protocol $2$, but with the roles of Alice and Bob exchanged.  
\end{enumerate}
\vspace{-1\baselineskip}
\hrulefill
\medskip

Observe that embedded within this protocol are $4 n (n-1)$ CHSH games.  In sub-protocol~$2$, for example, for every $i \neq j$ and $c \in \{0,1\}$, Alice and Bob are playing a CHSH game.  In general, Alice's strategy for the CHSH game can depend on~$i$, while Bob's strategy can depend on $i, j$ and~$c$.  

The key ingredient in the protocol is the ``dummy" question~$j$ used in parts $2$ and $3$.  Given the indices $\{i, j\}$, Bob does not know which index was given to Alice.  Intuitively, this forces Bob to measure two independent qubits for his responses $y$ and~$y'$.  In our analysis below, it allows us to establish the state-dependent commutation relationship between Alice's $i$th and $j$th qubits: $\max_{P, Q \in \{X, Z\}} \norm{[P_i, Q_j]_\Alice{A} \ket \psi} \approx 0$.  The dummy question is related to the technique of oracularization in complexity theory.  To our knowledge, this was first used in the context of two-player quantum games in~\cite{ItoKobayashiMatsumoto09oracularization}.  

Note that Alice cannot distinguish between the first and second sub-protocols, and Bob cannot distinguish between the first and third sub-protocols.  The role of the first sub-protocol is to establish consistency between Alice and Bob's $i$th qubits, without dependence on $(j, c)$.\footnote{The first sub-protocol could also be replaced with standard CHSH games.}  

Let us now analyze the completeness and soundness of this protocol.

\subsubsection{Completeness} \label{s:EPRtestingcompleteness}

Honest players Alice and Bob, sharing $n$ EPR states, $\ket{\EPRstate}{}^{\otimes n}$, play as follows: 

\begin{itemize}
\item 
If given a single tuple $(i, a)$, the player measures her share of the $i$th EPR state with $\sigma^z$ if $a = 0$, or with $\sigma^x$ if $a = 1$.  She returns the measured eigenvalue.  
\item 
If given two tuples $(i, b)$ and $(j, c)$, the player measures her share of the $i$th EPR state with $H = \tfrac{1}{\sqrt 2}\big(\begin{smallmatrix}1&1\\1&-1\end{smallmatrix}\big)$ if $b = 0$, or with $-\sigma^x H \sigma^x$ if $b = 1$.  She returns the measured eigenvalue.  To determine~$y'$, she uses the same rules to measure the $j$th EPR state depending on~$c$.  
\end{itemize}

\noindent
These strategies correspond to using the indicated qubits $i$ and~$j$ to play CHSH games following the optimal strategy from \secref{s:CHSHgame}, where the player given a single tuple takes the role of CHSH game player~$1$ and the player given two tuples takes the role of CHSH game player~$2$.  

When Alice and Bob follow this honest strategy, they pass sub-protocol $1$ with probability~$1$, and they pass sub-protocols $2$ and~$3$ with probability $\cos^2 \tfrac{\pi}{8}$, the optimal success probability for a CHSH game.  Overall the verifier accepts with probability 
\begin{equation} \label{e:omegaoptdef}
\omegaopt := \tfrac{1}{3}(1 + 2 \cos^2 \tfrac{\pi}{8}) \approx 0.90
 \enspace .
\end{equation}

\subsubsection{Soundness}

In general, a strategy for Alice and Bob consists of: 
\begin{itemize}
\item 
Finite-dimensional Hilbert spaces $\H_\Alice{A}, \H_\Bob{B}$, and a shared state $\ket \psi \in \H_\Alice{A} \otimes \H_\Bob{B}$.  
\item 
Reflections $(Z_i)_D, (X_i)_D \in \L(\H_D)$, for $D \in \{ \Alice{A}, \Bob{B} \}$.  On question $(i, 0)$, player~$D$ measures $(Z_i)_D$ and returns the measured eigenvalue~$\pm 1$; on question $(i, 1)$, she returns the result of measuring $(X_i)_D$.  Note that $(Z_i)_D$ and $(X_i)_D$ need not anti-commute.  
\item 
For each $D \in \{ \Alice{A}, \Bob{B} \}$, for every $i < j$ and $b, c \in \{0,1\}$, a complete set of orthogonal projections $\{ (\Pi^{(i,b),(j,c)}_{x,y})_D : x, y \in \{\pm 1\} \}$.  On receiving question $\big( (i, b), (j, c) \big)$, player~$D$ measures $\ket \psi$ with the projections and returns the results $x, y$.  
\end{itemize}
Naimark's theorem ensures the assumption of projective measurements is without loss of {generality}.  

We state the results of the soundness analysis as a theorem: 

\begin{theorem} \label{t:nEPRsoundness}
Consider a strategy that is accepted with probability at least $\omegaopt - \delta$, with $\delta > 0$.  Then there exist spaces $\H'_\Alice{A}$, $\H'_\Bob{B}$ extending $\H_\Alice{A}, \H_\Bob{B}$ respectively, and extensions of the $(Z_i)_D$ to $\H'_D$ by a direct sum with other reflections, and reflections $(\bar X_i)_D$ such that $\{ (Z_i)_D, (\bar X_i)_D \} = 0$, $\norm{ (X_i - \bar X_i)_D \ket \psi } = O(n \sqrt \delta)$, and 
\begin{align*}
\max_{\substack{i, j \in [n], \, i \neq j \, \\ P, Q \in \{\bar X, Z\}}} \! \bignorm{ [P_i, Q_j]_D \ket \psi } &= O(n \sqrt \delta) \\
\max_{i \in [n], P \in \{\bar X, Z\}} \bignorm{ (P_i)_\Alice{A} (P_i)_\Bob{B} \ket \psi - \ket \psi } &= O(n \sqrt \delta)
 \enspace .
\end{align*}
\end{theorem}

In particular, the conclusions of this theorem are exactly the assumptions required to apply \thmref{t:EPRstabilizersfrommath} and \corref{t:EPRpairsfromstabilizers}, with $\epsilon = O(n \sqrt \delta)$.  Thus, combining \thmref{t:nEPRsoundness} and \corref{t:EPRpairsfromstabilizers}, directly leads to the robustness guarantee $O(n^{5/2} \sqrt \delta)$ claimed in the introduction.  

\begin{proof}
For each $D \in \{ \Alice{A}, \Bob{B} \}$, for $i < j$, define reflections $(R^{jc}_{ib})_D$ by 
\begin{equation*}
R^{jc}_{ib} = \sum_{x \in \{\pm 1\}} x \cdot \Big( \Pi^{(i,b),(j,c)}_{x,1} + \Pi^{(i,b),(j,c)}_{x,-1} \Big)
 \enspace .
\end{equation*}
Thus $R^{jc}_{ib}$ is the reflection measured to determine the response to $(i, b)$, marginalized over the response to $(j, c)$.  Similarly, for $i > j$, define $R^{jc}_{ib} = \sum_{x, y \in \{\pm 1\}} y \, \Pi^{(j,c),(i,b)}_{x,y}$.  Then for $i \neq j$, 
\begin{equation} \label{e:commutingmeasurements}
\big[ R^{jc}_{ib}, R^{ib}_{jc} \big] = 0
 \enspace .
\end{equation}

Since the optimal probability of winning a CHSH game is $\cos^2 \tfrac{\pi}{8}$, the maximum probability that the verifier can accept in sub-protocols $2$ and~$3$ is also $\cos^2 \tfrac{\pi}{8}$.  In particular, by a Markov inequality, this implies that the probability that Alice and Bob win in sub-protocol~$1$, conditioned on the verifier having chosen~$i \in [n]$, is at least $1 - 3 n \delta$.  Hence, by a straightforward calculation, 
\begin{equation} \label{e:paulistabilizedEPR}
\max\Big\{ 
\bignorm{ (X_i)_\Alice{A} (X_i)_\Bob{B} \ket \psi - \ket \psi }, 
\bignorm{ (Z_i)_\Alice{A} (Z_i)_\Bob{B} \ket \psi - \ket \psi }
\Big\} \leq 2 \sqrt{3 n \delta}
 \enspace .
\end{equation}
Again by a Markov inequality, for each of the $4 n (n-1)$ CHSH games embedded in sub-protocols $2$ and~$3$, the success probability is at least $\cos^2 \tfrac\pi8 - 6 n (n-1) \delta$.  Therefore \lemref{t:eprlemma} applies to each game.  (The condition on each observable's two eigenspaces having the same dimension is easily satisfied by at most doubling the dimension of the Hilbert space.)  

Consider the game in sub-protocol~$2$ specified by $i \neq j$ and $c \in \{0,1\}$.  The CHSH game measurement operators are $(R_0)_1 = (Z_i)_\Alice{A}$, $(R_1)_1 = (X_i)_\Alice{A}$ and $(R_b)_2 = (R^{jc}_{ib})_\Bob{B}$ for $b = 0, 1$.  From \lemref{t:eprlemma} we obtain that there exists a reflection $(\bar X_i)_\Alice{A}$, depending only on~$i$ (not $j$ or~$c$), such that $\{ Z_i, \bar X_i \} = 0$ and $\norm{(X_i - \bar X_i)_D \ket \psi} = O(n \sqrt \delta)$.  In particular, from Eq.~\eqnref{e:paulistabilizedEPR} this implies that $\bignorm{ (\bar X_i)_\Alice{A} (\bar X_i)_\Bob{B} \ket \psi - \ket \psi } = O(n \sqrt \delta)$.  

Recall that in the ideal CHSH game strategy from \secref{s:CHSHgame}, $(R_b)_2 = \tfrac{1}{\sqrt 2}(\sigma^z + (-1)^b \sigma^x)$.  Also $\ket{\EPRstate}$ is fixed by $\tfrac12 (\sigma^z + (-1)^b \sigma^x)^{\otimes 2}$.  From \lemref{t:eprlemma}, the actual measurement operators have an action on $\ket \psi$ that is close to that of these ideal operators, and $\ket \psi \approx \ket{\EPRstate} \otimes \ket{\hat \psi}$.  Combining these bounds using triangle inequalities gives 
\begin{equation} \label{e:hadamardstabilizedEPR}
\bignorm{ \tfrac{1}{\sqrt 2}(Z_i + (-1)^b \bar X_i)_\Alice{A} ( R^{jc}_{ib} )_\Bob{B} \ket \psi - \ket \psi } = O(n \sqrt \delta)
 \enspace .
\end{equation}

We will use the bounds in Eq.~\eqnref{e:hadamardstabilizedEPR} to show that $[P_i, Q_j]_\Alice{A} \ket \psi \approx 0$ for $P, Q \in \{ \bar X, Z \}$ and $i\neq j$.  For $i \neq j$, we compute that for $b, c \in \{0,1\}$, 
\begin{align*}
\tfrac{1}{\sqrt 2}(Z_i + (-1)^b \bar X_i)_\Alice{A} \tfrac{1}{\sqrt 2}(Z_j + (-1)^c \bar X_j)_\Alice{A} \ket \psi
&\approx 
\tfrac{1}{\sqrt 2}(Z_i + (-1)^b \bar X_i)_\Alice{A} (R^{ib}_{jc})_\Bob{B} \ket \psi \\
&\approx 
(R^{ib}_{jc} R^{jc}_{ib})_\Bob{B} \ket \psi \\
&= 
(R^{jc}_{ib} R^{ib}_{jc})_\Bob{B} \ket \psi \\
&\approx 
\tfrac{1}{\sqrt 2}(Z_j + (-1)^c \bar X_j)_\Alice{A} (R^{jc}_{ib})_\Bob{B} \ket \psi \\
&\approx 
\tfrac{1}{\sqrt 2}(Z_j + (-1)^c \bar X_j)_\Alice{A} \tfrac{1}{\sqrt 2}(Z_i + (-1)^b \bar X_i)_\Alice{A} \ket \psi
 \enspace ,
\end{align*}
where each approximation is up to error $O(n \sqrt \delta)$ in Euclidean distance, by Eq.~\eqnref{e:hadamardstabilizedEPR}, and the middle equality is by Eq.~\eqnref{e:commutingmeasurements}.  Thus we have $\bignorm{\big[Z_i + (-1)^b \bar X_i, Z_j + (-1)^c \bar X_j \big]_\Alice{A} \ket \psi} = O(n \sqrt \delta)$.  Taking sums and differences of these bounds, it follows that 
\begin{equation*}
\max\!\big\{ \norm{[Z_i, Z_j]_\Alice{A} \ket \psi}, \norm{[Z_i, \bar X_j]_\Alice{A} \ket \psi}, \norm{[\bar X_i, Z_j]_\Alice{A} \ket \psi}, \norm{[\bar X_i, \bar X_j]_\Alice{A} \ket \psi} \big\} = O(n \sqrt \delta)
 \enspace .
\end{equation*}
The bound $\max_{P, Q \in \{\bar X, Z\}} \norm{[P_i, Q_j]_\Bob{B} \ket \psi} = O(n \sqrt \delta)$ follows by symmetry.  
\end{proof}

\subsection{An attack: lower bound on robustness of the protocol} \label{s:robustEPRprotocol}

The soundness analysis from \thmref{t:nEPRsoundness}, \thmref{t:EPRstabilizersfrommath} and \corref{t:EPRpairsfromstabilizers} implies that any strategy with success probability $\omegaopt - 1/n^{O(1)}$ in our protocol must have local dimension at least $2^n$ for each player.  The following construction shows that the same conclusion cannot be extended to strategies having success probability below $\omegaopt - c \sqrt{(\log n) / n}$, for some constant~$c$.  

\begin{lemma}
For any $\epsilon > 0$ there exists a strategy for the players that is accepted with probability at least $\omegaopt - \epsilon$ in the protocol for testing $n$ qubits, but such that the players' Hilbert spaces $\H_\Alice{A}$ and $\H_\Bob{B}$ have dimension only $2^{O(\log n/\epsilon^2)}$.  
\end{lemma}

\begin{proof}[Proof sketch.]
We sketch the argument, which is based on a dimension-reduction construction from~\cite[Theorem 3.1]{ChaoReichardtSutherlandVidick16overlapping}.  The theorem states that there exist $n$ pairs of anti-commuting reflections $(X_i, Z_i)$ in a space of dimension $d = 2^{O(\log n/\epsilon^2)}$, satisfying $\norm{[P_i, Q_j]} \leq c \, \epsilon$ for $i \neq j$ and $P, Q$ either $X$ or~$Z$, where $c > 0$ is a constant that we can choose.  

Assume that Alice and Bob share a $d$-dimensional maximally entangled state and play the protocol in \secref{s:EPRpairprotocol} as follows:   
\begin{enumerate}
\item 
When asked to play a CHSH game on qubit $i \in [n]$, a player follows the ideal CHSH game strategy using the observables $Z_i, X_i$.  Thus the players pass the first sub-protocol with probability one.  
\item 
When asked to play CHSH games on qubits $i, j$, with $i < j$, a player sequentially measures using $H_i = (X_i+Z_i)/\sqrt{2}$ or $-X_i H X_i$, and then similarly for $j$.  Say for example that in sub-protocol~$2$ Bob is given the indices $i < j$.  If Alice is given index~$i$, then the players win the CHSH game with the optimal probability, $\cos^2 \tfrac\pi8$.  If Alice is given index~$j$, then Bob's first measurement using the observables associated with qubit~$i$ can slightly perturb the state on qubits~$j$.  However, provided~$c$ is chosen small enough, the commutation bounds imply that the distribution of outcomes is $\epsilon$-close in statistical distance from the case in which Bob measures qubit~$j$ first.  Thus the players succeed with probability at least $\cos^2 \tfrac\pi8 - \epsilon$ in the second and third sub-protocols.  \qedhere
\end{enumerate}
\end{proof}

\subsection*{Acknowledgements}

R.C., B.R.~and C.S.~supported by NSF grant CCF-1254119 and ARO grant W911NF-12-1-0541.  T.V.~supported by NSF CAREER grant CCF-1553477, AFOSR YIP award number FA9550-16-1-0495, and the IQIM, an NSF Physics Frontiers Center (NFS Grant PHY-1125565) with support of the Gordon and Betty Moore Foundation (GBMF-12500028).

\appendix

\section{Separation of $n$ qubits from any entangled state} \label{s:psiselftest}

Based on the results of~\cite{yangnavascues13robust,BampsPironio15sos} it is possible to extend the results from \secref{s:EPRselftest} to \mbox{separating}~$n$ partially overlapping qubits based on the use of an arbitrary two-qubit entangled state $\ket{\pst} = \cos\theta \ket{00} + \sin\theta \ket{11}$, for $\theta \in (0,\pi/2)$.  The arguments are very similar, replacing the CHSH game by a game based on a Bell inequality introduced in~\cite{yangnavascues13robust,BampsPironio15sos}.  In the following two sections we sketch the arguments for extending Theorems~\ref{t:EPRstabilizersfrommath} and~\ref{t:nEPRsoundness}, respectively, to this scenario.

\subsection{Separating $n$ qubits} \label{sec:s-multiple}

The following is an analogue of \thmref{t:EPRstabilizersfrommath}.  

\begin{theorem} \label{t:psifrommath}
Let $\ket \psi \in \H_\Alice{A} \otimes \H_\Bob{B}$, with $\norm{\ket \psi} = 1$.  Assume that for $j \in [n]$ we are given reflections 
\begin{equation*}
(X_j)_\Alice{A}, (Z_j)_\Alice{A} \in \L(\H_\Alice{A})
\quad \text{and} \quad
(X_j)_\Bob{B}, (Z_j)_\Bob{B} \in \L(\H_\Bob{B})
\end{equation*}
such that the following conditions hold: 
\begin{align} 
&\{(X_j)_D, (Z_j)_D\} = 0 
&\forall D \in \{\Alice{A}, \Bob{B}\}, \, \forall  j \in [n] \label{eq:psi-assumptions1} \\
&\bignorm{ (Z_j)_\Alice{A} \otimes (Z_j)_\Bob{B} \ket \psi - \ket \psi } \leq \epsilon 
&\forall j \in [n] \label{eq:psi-assumptions2} \\ 
\bigg\|\begin{split}
\sin\theta \, (X_j)_D \otimes (\Id + (Z_j)_{D'}) \ket \psi \\
- \cos\theta \, (\Id - (Z_j)_D)\otimes (X_j)_{D'}  \ket \psi 
\end{split}\bigg\| \leq \epsilon
&\forall D \neq D' \in \{\Alice{A}, \Bob{B}\}, \, \forall j \in [n] \label{eq:psi-assumptions3} \\ 
&\bignorm{ [P_i, Q_j]_D \ket \psi } \leq \epsilon 
&\forall D \in \{\Alice{A}, \Bob{B}\}, \, \forall P, Q \in \{X,Z\}, \, \forall i \neq j \in [n] \label{eq:psi-assumptions4}
\end{align}
Let 
\begin{equation*}
\ket{\psi'} = \ket \psi \otimes \ket{\pst}^{\otimes n}_\Alice{A'} \otimes \ket{\pst}^{\otimes n}_\Bob{B'} \otimes \ket{00}_{\Alice{A''}\Bob{B''}}^{\otimes n} \in \H_\Alice{A} \otimes (\C^2)^{\otimes 2n}_\Alice{A'} \otimes (\C^2)^{\otimes n}_\Alice{A''} \otimes \H_\Bob{B} \otimes (\C^2)^{\otimes 2n}_\Bob{B'} \otimes (\C^2)^{\otimes n}_\Bob{B''}
 \enspace .
\end{equation*}
Then there exist reflections $X_1', Z_1', \ldots, X_n', Z_n'$ on $\H_D \otimes (\C^2)^{\otimes 2n}_{D'}$, with $\{X_j', Z_j'\} = 0$ and for $P, Q \in \{X,Z\}$, $[P_i', Q_j'] = 0$ for $i \neq j$, $\bignorm{ \big( P_j' - P_j \otimes \identity_{D'} \big) \ket{\psi'} } = O(n \epsilon)$, and 
\begin{equation*}
\bignorm{ (P_j')_\Alice{AA'} \otimes (P_j')_\Bob{BB'} \ket{\psi'} - \ket{\psi'} } = O(n \epsilon)
 \enspace .
\end{equation*}
\end{theorem}

Note that conditions~\eqnref{eq:psi-assumptions2} and~\eqnref{eq:psi-assumptions3} replace the conditions used to certify EPR states in \thmref{t:EPRstabilizersfrommath}.  The following analog of \claimref{t:EPRstatestabilized} justifies these conditions by showing that they characterize the state $\ket\pst$: 

\begin{claim}[\cite{yangnavascues13robust,BampsPironio15sos}] \label{t:psistabilized}
For any $0 < \theta < 1$ there exists a constant~$c_\theta$ such that if $\ket \phi \in \C^2 \otimes \C^2 \otimes \H$ is a state satisfying 
\begin{equation*}
\max\big\{ \norm{\sigma^z_1 \sigma^z_2 \ket \phi - \ket \phi}, \norm{\sin\theta \, \sigma^x_1(I + \sigma^z_2) \ket \phi - \cos\theta \, \sigma^x_2 (I - \sigma^z_2) \ket \phi} \big\} \leq \delta
 \enspace ,
\end{equation*}
then there exists  a state $\ket{\phi'} \in \H$ such that 
\begin{equation*}
\bignorm{ \ket \phi - \ket{\pst} \otimes \ket{\phi'} } \leq c_\theta\,\delta \enspace .
\end{equation*}
\end{claim}

Based on the claim we can obtain the following corollary to \thmref{t:psifrommath}.  

\begin{corollary} \label{c:psifrommath}
Under the assumptions of \thmref{t:psifrommath}, there are unitaries $U_D : \H_D \otimes (\C^2)^{\otimes 2n}_{D'} \otimes (\C^2)^{\otimes n}_{D''} \rightarrow (\C^2)^{\otimes n}_D \otimes \hat \H_D$, for $D \in \{\Alice{A}, \Bob{B}\}$, and a state $\ket{\ancilla} \in \hat \H_\Alice{A} \otimes \hat \H_\Bob{B}$ such that 
\begin{align*}
\bignorm{ U_\Alice{A} \otimes U_\Bob{B} \ket{\psi'} - \ket{\pst}_{\Alice{A}\Bob{B}}^{\otimes n} \otimes \ket{\ancilla} } &= O(n^{3/2} \epsilon) \\
\norm{ (U X_j U^\dagger - \sigma^x_j \otimes \identity_{\hat \H})_D \ket{\pst}_{\Alice{A}\Bob{B}}^{\otimes n} \otimes \ket{\ancilla} } &= O(n^{3/2} \epsilon) \\
\norm{ (U Z_j U^\dagger - \sigma^z_j \otimes \identity_{\hat \H})_D \ket{\pst}_{\Alice{A}\Bob{B}}^{\otimes n} \otimes \ket{\ancilla} } &= O(n^{3/2} \epsilon)
 \enspace .
\end{align*}
\end{corollary}

\begin{proof}[Proof sketch of \thmref{t:psifrommath}]
The proof follows closely the proof of \thmref{t:EPRstabilizersfrommath}.  The main result needed is the existence of appropriate SWAP operators $\mathcal{S}_j$ satisfying the properties of Lemmas~\ref{t:pullswapfromAlicetoBob} and~\ref{t:pauliswapcommuteondifferentqubits}.  

\def\control #1#2{{\text{$\mathrm{CTL}_{#1}$-${#2}$}}}

For $j \in [n]$ let
\begin{equation} \label{eq:t-def-u}
(\mathcal{U}_j)_{\Alice{A}\Alice{A''_j}} = (H_j)_{\Alice{A''_j}} (\control{\Alice{A''_j}}{(Z_j)_\Alice{A}}) (H_j)_\Alice{A''_j} (\control{\Alice{A''_j}}{(X_j)_\Alice{A}})
 \enspace ,
\end{equation}
where $(H_j)_\Alice{A''}$ is a Hadamard on the ancilla qubit and $\control{\Alice{A''}}{(Z_j)_\Alice{A}}$ (resp.~$\control{\Alice{A''}}{(X_j)_\Alice{A}}$) is $(Z_j)_\Alice{A}$ (resp.~$(X_j)_\Alice{A}$), controlled on the qubit in $\Alice{A''}$.  Define $\mathcal{U}_{\Bob{B}\Bob{B''}}$ similarly.  As shown in~\cite{yangnavascues13robust,BampsPironio15sos}, it follows from~\eqnref{eq:psi-assumptions2} and~\eqnref{eq:psi-assumptions3} that for $D \in \{\Alice{A}, \Bob{B}\}$ and $P \in \{X, Z\}$, if $(P'_j)_{DD_j''} = (\mathcal{U}_j^\dagger)_{DD''_j} (\Id_{D} \otimes (P_j)_{D''_j} )(\mathcal{U}_j)_{DD''_j}$ then 
\begin{equation} \label{eq:t-u-pst}
 \big\| (\mathcal{U}_j)_{\Alice{AA''_j}} \otimes (\mathcal{U}_j)_{\Bob{BB''_j}} \ket{\psi'} - \ket{\pst}_{\Alice{A''_j}\Bob{B''_j}} \otimes \ket{\mathrm{junk}}_{\Alice{A}\Bob{B}} \otimes \ket{\pst}_{\Alice{A'}}^{\otimes n}\otimes \ket{\pst}_{\Bob{B'}}^{\otimes n} \big\| = O(\epsilon)
\end{equation}
for some state $\ket{\mathrm{junk}}_{\Alice{A}\Bob{B}}$, and 
\begin{equation*}
\big\| ( (P_j)_{D} -  (P'_j)_{DD''})\ket{\psi'} \big\| = O(\epsilon)
 \enspace .
\end{equation*}
We now define $\mathcal{S}_j$ as 
\begin{equation*}
\mathcal{S}_j = (\mathcal{U}_j)_{\Alice{AA''}}^\dagger \SWAP_\Alice{A''A'} (\mathcal{U}_j)_\Alice{AA''}
 \enspace ,
\end{equation*}
and similarly define $\mathcal{S}_j'$, using $(\mathcal{U}_j)_\Bob{BB''}$ and swapping $\Bob{B''}$ with $\Alice{A'}$.  
The following lemma establishes the analog of Lemmas~\ref{t:pullswapfromAlicetoBob} and~\ref{t:pauliswapcommuteondifferentqubits}.  

\begin{lemma} \label{lem:psiswap}
For $i \in [n]$, 
\begin{equation} \label{eq:swap-switch}
\norm{ (\S_{i} - \S'_i) \ket{\psi'} } = O(\epsilon)
 \enspace .
\end{equation}
For $i, j \in [n]$, $i \neq j$, and $P \in \{X,Z\}$, 
\begin{equation} \label{eq:swap-t-com}
\bignorm{ [\S_i,(P'_j)_\Alice{AA''} \otimes \Id_\Alice{A'}] \ket{\psi'} } = O(\epsilon)
 \enspace .  
\end{equation}
\end{lemma}

\begin{proof}
For~\eqnref{eq:swap-switch}, we use~\eqnref{eq:t-u-pst}.  For~\eqnref{eq:swap-t-com}, it suffices to check that for $i \neq j$ the unitaries $(\mathcal{U}_i)_\Alice{AA''}$ and $(\mathcal{U}_j)_\Alice{AA''}$ approximately commute.  This follows from the definition~\eqnref{eq:t-def-u} and the approximate commutation conditions~\eqnref{eq:psi-assumptions4}.  
\end{proof}

Once \lemref{lem:psiswap} has been established, the proof of the theorem follows exactly the same steps as the proof of \thmref{t:EPRstabilizersfrommath}.  
\end{proof}

\subsection{The testing protocol} \label{s:sbasic}

In order to experimentally verify that the conditions of \thmref{t:psifrommath} are satisfied, a very similar protocol to the one described in \secref{s:EPRpairprotocol} can be used, except the CHSH tests should be replaced by the appropriate self-test for the state $\ket \pst$.  Such a test is developed in~\cite{yangnavascues13robust,BampsPironio15sos}, based on a Bell inequality with two inputs and two outputs per site.  

\begin{theorem}[\cite{yangnavascues13robust,BampsPironio15sos}] \label{thm:psi-test}
For $0 < \alpha < 2$ let $\mathcal{B}_\alpha = \alpha A_0 + A_0 (B_0 + B_1) + A_1 (B_0 - B_1)$, and let $\theta$ be such that $\tan \theta = \sqrt{(4 - \alpha^2)/(2 \alpha^2)}$.  Then the optimal violation of $\mathcal{B}_\alpha$ is $b_\alpha = \sqrt{8+2\alpha^2}$, and this violation can be achieved using state $\ket \pst$.  

Furthermore, suppose arbitrary observables $A_0,A_1$ and $B_0,B_1$ on $\H_\Alice{A}$ and $\H_\Bob{B}$ respectively lead to a violation of $b_\alpha - \epsilon$ when applied to a state $\ket{\psi}_{\Alice{A}\Bob{B}}$.  Assume that $A_0, A_1$ and $B_0, B_1$ each have both $\pm 1$ eigenspaces of the same dimension. Then there exists observables $P_\Alice{A}$ and $P_\Bob{B}$, for $P \in \{X, Z\}$, such that 
\begin{align}
&\bignorm{ (Z_\Alice{A} \otimes \Id_\Bob{B} - \Id_\Alice{A} \otimes Z_\Bob{B}) \ket \psi } =  O(\epsilon^{1/2}) \label{eq:s-z-close} \\
&\bignorm{ \sin\theta \, (X_\Alice{A} \otimes (\Id_\Bob{B} + Z_\Bob{B} ) \ket \psi - \cos\theta \, (\identity_\Alice{A} -  Z_\Alice{A}) \otimes X_\Bob{B} \ket \psi } = O(\epsilon^{1/2}) \label{eq:s-x-close-1}
 \enspace ,
\end{align}
where the $O(\cdot)$ notation hides factors depending on $\cos^{-1}\theta$ and $\sin^{-1}\theta$. 
Moreover, we can take $X_\Alice{A} = A_0$ and $Z_\Alice{A} = A_1$.  
\end{theorem}

Based on the Bell inequality $\mathcal{B}_\alpha$ it is possible to design a simple game between players ``Alice'' and ``Bob'' such that the player's maximum success probability in the game is directly related to the expectation value in $\mathcal{B}_\alpha$ that is induced by their strategy. 

We can then use exactly the same protocol as described in \secref{s:EPRpairprotocol}, with the CHSH game replaced by the game based on $\mathcal{B}_\alpha$.  The proof that the appropriate relations, as required in \thmref{t:psifrommath}, are satisfied by any strategy that is accepted in the protocol with probability close to the optimum then follows closely the proof of \thmref{t:nEPRsoundness} and we only sketch the argument.  

Observables $(P_i)_D$ are obtained directly from \thmref{thm:psi-test}.  Conditions~\eqnref{eq:psi-assumptions2} and~\eqnref{eq:psi-assumptions3} then \mbox{follow} from~\thmref{thm:psi-test}.  The commutation conditions~\eqnref{eq:psi-assumptions4} are proven exactly as in the proof of \mbox{\thmref{t:nEPRsoundness}}.  

The anti-commutation relations require a little more work.  From observables $X_D,Z_D$ define $\hat{Z}_\Alice{A} = Z_\Alice{A}$, $\hat{Z}_\Bob{B} = Z_\Bob{B}$, and 
\begin{equation*}
\hat{X}_\Alice{A} = \frac{\cos\theta}{2 \sin\theta} X_\Alice{A} (\identity_\Alice{A} - Z_\Alice{A}) - \frac{\sin\theta}{2 \cos\theta} X_\Alice{A} (\identity_\Alice{A} - Z_\Alice{A}) , \quad \hat{X}_\Bob{B} = \frac{\cos\theta}{2 \sin\theta} X_\Bob{B} (\identity_\Bob{B} - Z_\Bob{B}) - \frac{\sin\theta}{2 \cos\theta} X_\Bob{B} (\identity_\Bob{B} - Z_\Bob{B}) 
  \enspace .
\end{equation*}
It follows from~\eqnref{eq:s-z-close} and~\eqnref{eq:s-x-close-1} that for $T \in \{X, Z\}$, $\hat T$ is an operator of norm $O(\tan \theta + \tan^{-1} \theta)$ such that 
\begin{equation}
\max\Big\{ \bignorm{ (T_\Alice{A} \otimes \identity_\Bob{B} - \identity_\Alice{A} \otimes \hat{T}_\Bob{B}) \ket \psi },\, \bignorm{ (\identity_\Alice{A} \otimes T_\Bob{B} - \hat{T}_\Alice{A} \otimes \identity_\Bob{B}) \ket \psi } \Big\} = O\big( (\sin^{-1}\theta + \cos^{-1}\theta) \epsilon \big)
 \enspace .
\end{equation}
Using this it is then easy to verify from the same equations that approximate anti-commutation relations 
\begin{equation}
\max\Big\{ \bignorm{ \{ X_\Alice{A}, Z_\Alice{A}\} \otimes \identity_\Bob{B} \ket \psi },\, \bignorm{ \identity_\Alice{A} \otimes \{X_\Bob{B}, Z_\Bob{B}\} \ket \psi } \Big\} = O\big(( \sin^{-1} \theta + \cos^{-1} \theta) \epsilon \big)
\end{equation}
hold, as desired.  It is then not hard to turn such pairwise approximate anti-commutation into exact anti-commutation; see, e.g.,~\cite[Lemma~7]{OstrevVidick16entanglement}.

\bibliographystyle{alpha-eprint}
\bibliography{q}

\end{document}